\documentclass[twocolumn,prl,english,superscriptaddress,floatfix,longbibliography]{revtex4-2}
\usepackage[utf8]{inputenc}
\usepackage[total={6in, 8in}]{geometry}

\usepackage{amsmath}
\usepackage{hyperref,cleveref}
\usepackage{amssymb}
\usepackage{amsthm}
\usepackage{graphicx}
\usepackage{algorithmic}
\usepackage{mathrsfs}
\usepackage{braket}
\usepackage{xcolor}
\usepackage{tikz}
\usepackage[roman]{complexity}
\usetikzlibrary{quantikz}

\newtheorem{thm}{\protect\theoremname}
\theoremstyle{plain}
\newtheorem*{thm*}{\protect\theoremname}
\theoremstyle{plain}
\newtheorem{lem}[thm]{\protect\lemmaname}
\theoremstyle{plain}

\theoremstyle{plain}
\newtheorem*{lem*}{\protect\lemmaname}
\theoremstyle{plain}

\theoremstyle{plain}
\newtheorem{cor}[thm]{\protect\corollaryname}
\newtheorem{assumption}[thm]{Assumption}
\newtheorem{defn}[thm]{Definition}

\crefname{thm}{theorem}{theorems}
\crefname{cor}{corollary}{corollaries}
\crefname{lem}{lemma}{lemmas}

\newcommand\sectionprl[1]{\paragraph{\textit{#1}.---\!\!\!\!}}

\usepackage[USenglish]{babel}
  \providecommand{\corollaryname}{Corollary}
  \providecommand{\lemmaname}{Lemma}
  \providecommand{\propositionname}{Proposition}
  \providecommand{\remarkname}{Remark}
\providecommand{\theoremname}{Theorem}

\usepackage[normalem]{ulem}

\newcommand{\Or}{\mathcal{O}}
\newcommand{\RR}{\mathbb{R}}

\newcommand{\ZZ}{\mathbb{Z}}
\newcommand\QN{\lambda}

\newcommand\minmax[3]{\inf_{\phi_0,\ldots,\phi_{#2}#1}\sup\Big\{\braket{\psi|#3|\psi}\Big|\ket\psi\in\operatorname{span}\{\phi_0,\ldots,\phi_{#2}\},\|\ket\psi\|=1\Big\}}
\newcommand\essspec{\sigma_{\operatorname{ess}}}

\newcommand{\shrink}{\sigma}

\newcommand\comment[1]{}

\begin{document}
\title{Entanglement area law for 1D gauge theories and bosonic systems}
\author{Nilin Abrahamsen}\thanks{These two authors contributed equally to this work.}
\affiliation{Simons Institute and Department of Mathematics, University of California, Berkeley, CA 94720, USA}
\author{Yu Tong}\thanks{These two authors contributed equally to this work.}\affiliation{Department of Mathematics, University of California, Berkeley, CA 94720, USA}
\author{Ning Bao}\affiliation{Computational Science Initiative, Brookhaven National Laboratory, Upton, NY 11973 USA}
\author{Yuan Su}\affiliation{Google Quantum AI, Venice, CA 90291, USA}
\author{Nathan Wiebe}\affiliation{Department of Computer Science, University of Toronto, Toronto, ON M5S 3E1, Canada\\Pacific Northwest National Laboratory, Richland, WA 99354, USA}

\begin{abstract}
    We prove an entanglement area law for a class of 1D quantum systems involving infinite-dimensional local Hilbert spaces. This class of quantum systems includes bosonic models and lattice gauge theories in one spatial dimension. Our proof relies on new results concerning the robustness of the ground state and spectral gap to the truncation of Hilbert space, applied within the approximate ground state projector (AGSP) framework. Our result provides theoretical justification for using tensor networks to study the ground state properties of quantum systems with infinite local degrees of freedom.
\end{abstract}
\maketitle

\sectionprl{Introduction}
\label{sec:intro}

It has long been conjectured that for a wide range of quantum systems described by gapped local Hamiltonians, the entanglement entropy with respect to any bipartition of the system scales as the boundary area. Such entanglement entropy scaling is known as the \emph{entanglement area law}. An entanglement area law for the ground state of 1D quantum spin systems was first proved in the seminal paper by Hastings \cite{Hastings2007area}, and the scaling with respect to the spectral gap was improved by later work \cite{ALV12,AradKitaevLandauVazirani2013area}. Entanglement area laws have also been proved for degenerate ground states \cite{Chubb_2016,ALVV17} and low-lying eigenstates \cite{ALVV17}. Limited results are also available for higher-dimensional quantum systems, especially for the case where the Hamiltonian is frustration-free \cite{anshu2020entanglement,anshu2021area}. For 1D systems, whether a quantum state satisfies an entanglement area law is an important criterion for determining whether it can be approximated by a matrix-product state \cite{fannes1992finitely,Vidal}, which is a key component in the density-matrix renormalization group (DMRG) algorithm \cite{white1993density,Gharibian2015Hamiltonian}. 
Using theoretical tools constructed for proving the 1D area law, polynomial-time algorithms for computing the ground state of 1D gapped local Hamiltonians were given 
\cite{landau2015polynomial,ALVV17}.

The aforementioned results are all proved in the setting of quantum spin systems, i.e., each lattice site is associated with a finite-dimensional local Hilbert space. 
However, there are many quantum systems of practical interest that involve \emph{infinite-dimensional} local Hilbert spaces. Examples of such include bosonic systems and gauge theories. When a quantum system involves bosons, each bosonic mode corresponds to an infinite-dimensional local Hilbert space representing the occupation number of the mode. A similar situation arises when we consider lattice gauge theories (LGTs), with the Hamiltonians constructed according to \cite{kogut1975hamiltonian}. Given a fixed lattice discretization of a gauge theory, each gauge link (an edge of the lattice) has a local Hilbert space that is spanned by all the elements of the symmetry group, which is infinite-dimensional when the symmetry group contains infinitely many elements. The gauge theories that are of the greatest interests, i.e., the U(1), SU(2), and SU(3) theories, all fall into this category.

Tensor network methods have been extensively applied to studying LGTs to obtain interesting numerical results \cite{PichlerDalmonteEtAl2016real,BanulsBlattEtAl2020simulating,BanulsCichyEtAl2017efficient,banuls2013mass,banuls2013matrix,silvi2017finite,papaefstathiou2021density,tagliacozzo2014tensor}. The entanglement area law is a prerequisite for the ground state to be efficiently approximable by a tensor network state, and hence our result for LGTs provides a theoretical foundation for these numerical results from previous work.

There are two ways in which standard area law techniques are insufficient for our current setting: 
Firstly, the known area laws exhibit a bound which depends on the dimension of the local Hilbert space. The state-of-the-art 1D area law result bounds the entanglement entropy as $\Or(\Delta^{-1}\log^3(d))$ where $d$ is the local Hilbert space dimension and $\Delta$ is the spectral gap  \cite{AradKitaevLandauVazirani2013area,ALVV17}. 
This becomes infinity when the local Hilbert space is infinite-dimensional. 
Secondly, in quantum spin systems all local Hamiltonian terms can be rescaled to have operator norm at most $1$, whereas in the models we consider in this work the local Hamiltonian terms can be unbounded which precludes such a normalization.

For certain non-interacting bosonic systems involving infinite-dimensional local Hilbert spaces the entanglement area law has been proven, for example when the system is exactly solvable \cite{AudenaertEisertPlenioWerner2002entanglement,PlenioEisertDreissigCramer2005entropy,CramerEisertEtAl2006entanglement,CramerEisert2006correlations,CramerEisertPlenio2007statistics,Eisert2010arealaw}.  
However, a general methodology is unavailable for establishing area laws for quantum systems with infinite local degrees of freedom.

This letter gives an entanglement area law for a class of 1D quantum systems that involve bosons or arise from gauge theories, using the approximate ground state projector (AGSP) framework developed in \cite{ALV12}. Our quantum systems of interest involve infinite-dimensional local Hilbert spaces. Examples of such quantum systems include the Hubbard-Holstein model, U(1) and SU(2) LGTs, all of which are defined on a 1D chain. For these models, we can introduce a notion of \textit{local quantum number}, which is the occupation number in the bosonic case, the electric field value in the U(1) LGT, and the total angular momentum in the SU(2) LGT.

\sectionprl{The abstract model and main result}
\label{sec:the_abstract_model}
We consider a system on a line of length $N+1$ with a geometrically local Hamiltonian of the form
$H=H_1+H_2+\cdots+H_N$, where $H_x$ acts on sites $x-1$ and $x$.
At each site $x$ we have local observable $\QN_x$, which we call the local quantum number. The conditions for our results are stated in terms of the following quantities.

\begin{defn}
\label{defn:local}
\leavevmode\begin{enumerate}
\item
Let $\Pi_S^{(x)}=1_{S}(\QN_x)$ be the spectral projector for $\QN_x$ corresponding to eigenvalues in the set $S$. 
    \item For cutoff $\Lambda\ge0$ define the \emph{truncated local Hilbert space dimension} of a site $x$ to be $d(\Lambda)=\mathrm{rank}(\Pi_{[-\Lambda,\Lambda]}^{(x)})$. 
    \item For cutoff $\Lambda\ge0$ define the truncated norm of the local Hamiltonian $H_x$ constrained to $[-\Lambda,\Lambda]$ to be $\mathcal{N}(\Lambda)=\max_x\|H_x \Pi^{(x)}_{[-\Lambda,\Lambda]}\| $ where $\|\cdot\|$ is the spectral norm and the maximum is over sites $x$. 
\end{enumerate}
\end{defn}
We require that  the truncated local norm and Hilbert space dimension of the local Hamiltonian grow at most polynomially with with the cutoff $\Lambda$.  
\begin{assumption}
\label{assumption:local_quantum_numbers}  The  
maximum truncated local dimension $d(\Lambda)$ and norm $\mathcal{N}(\Lambda)$ 
satisfy
$
d(\Lambda),\ \mathcal{N}(\Lambda)=\Or(\poly(\Lambda)).
$
Also $\braket{|\lambda_x|}=\Or(1)$ where $\braket{\cdot}$ denotes the ground state expectation value.
\end{assumption}
Following \cite{TongAlbertMcCleanPreskillSu2021provably} we also assume a site-dependent decomposition of the global Hamiltonian into a quantum number-modifying part $H_W$ and a quantum number-preserving part $H_R$. As will be explained later, these assumptions are satisfied by a variety of quantum systems, to which our result applies.
\begin{assumption}[Growth of local quantum numbers]
\label{assumption:local_quantum_number_growth}
There exist non-negative real-valued constants $\chi$ and $r$ such that, for any $x$, the Hamiltonian can be decomposed as
\begin{equation}
    \label{eq:ham_walk}
    H = H^{(x)}_W + H^{(x)}_R,
\end{equation}
where $H^{(x)}_W$ and $H^{(x)}_R$ satisfy (for $\Pi_{\lambda}^{(x)}\equiv \Pi_{\{\lambda\}}^{(x)}$)
\begin{subequations}
\label{eq:m_conditions_ham_general}
\begin{align}
    &\Pi^{(x)}_{\lambda}H_W^{(x)} \Pi^{(x)}_{\lambda'}=0,\ \text{if }|\lambda-\lambda'|>1, \label{eq:m_conditions_ham_general_1}\\
    &\|H_W^{(x)}\Pi^{(x)}_{[-\Lambda,\Lambda]}\|\leq \chi(\Lambda+1)^{r}, \label{eq:m_conditions_ham_general_2}\\
    &[H_R^{(x)},\Pi^{(x)}_{\lambda}]=0~\text{for all~}\lambda, \label{eq:m_conditions_ham_general_3}
\end{align}
for all $x=1,2,\ldots,N$.
\end{subequations}
\end{assumption}
\begin{thm*}[Main result]
For the gapped ground state of any Hamiltonian $H$ satisfying Assumptions \ref{assumption:local_quantum_numbers} and \ref{assumption:local_quantum_number_growth}, in particular the 1D Hubbard-Holstein model, or the 1D U(1) or SU(2) LGTs, the entanglement entropy across a cut scales as $\Or(\poly(\Delta^{-1}))$, where $\Delta$ is the spectral gap, assuming that all coefficients in the Hamiltonian remain constant. In particular, this entanglement entropy scaling is independent of the system size.
\end{thm*}

\sectionprl{The proof strategy}
Our result utilizes a local quantum number tail bound recently obtained in \cite{TongAlbertMcCleanPreskillSu2021provably}. 
This tail bound tells us that for the relevant class of quantum systems, a spectrally isolated energy eigenstate can be well-approximated by a truncated state with low local quantum numbers.
This suggests that we may choose an effective local Hilbert space with finite dimensions and thereby apply the spin chain framework developed in \cite{ALV12,AradKitaevLandauVazirani2013area}.
To implement this proof strategy, we need to apply truncation to the local Hilbert space and by extension to the Hamiltonian terms, on each lattice site close to the cut. We show in Theorem~\ref{thm:robustness} that the truncation changes the ground state only by an exponentially small amount. This is similar to what is known as the robustness theorem \cite[Theorem 6.1]{AradKitaevLandauVazirani2013area}. Because of this we can show that the AGSP framework developed in \cite{AradKitaevLandauVazirani2013area} is robust to this truncation. In the case of frustration-free spin chains Ref.\ \cite{ALV12} showed that the existence of an AGSP with exact target space and suitable parameters directly implies an area law. Analogously, the ``off-the-shelf'' lemma with approximate target space in \cite{OTR} incorporates techniques from \cite{ALVV17} to directly conclude an area law from the existence of a sequence of AGSP with increasingly accurate target spaces. Using this off-the-shelf lemma the truncation-based AGSP implies the entanglement area law for our class of unbounded quantum systems. 

In the course of proving the area law for unbounded systems we show that the mean absolute value of the local quantum number can be bounded independently of the system size for many quantum systems without translation symmetry (Lemmas \ref{lem:local_quantum_number_bound_first_step} and \ref{lem:local_quantum_number_bound_general}), a bound that we consider to be of independent interest. These systems include the Hubbard-Holstein model (Corollary \ref{cor:mean_abs_bound_Hubbard_Holstein}) and U(1) and SU(2) LGTs (Corollary \ref{cor:mean_abs_bound_LGT}) with open boundary conditions.

In proving area laws for unbounded quantum systems, it is tempting to plug the eigenstate tail bound from \cite{TongAlbertMcCleanPreskillSu2021provably} directly into the area law result for spin systems \cite{AradKitaevLandauVazirani2013area}. However, such a naive strategy does \emph{not} seem to work here.
The main reason is that the tail bound from \cite{TongAlbertMcCleanPreskillSu2021provably} only guarantees the proximity of the quantum states before and after truncation, but not the proximity of the corresponding entanglement entropies. In the finite-dimensional setting, one could invoke the Fannes' inequality to estimate the entanglement entropy difference \cite{watrous2018theory,wilde2013quantum}, but its explicit dependence on the dimension of the Hilbert space will ruin the area law scaling for unbounded quantum systems. In fact, the main contribution of our paper is to give a careful truncation of the local Hilbert space that does not blow up entanglement entropy of the ground state, overcoming the technical difficulties mentioned above.

\sectionprl{Application to 1D U(1) LGT}
\label{sec:U(1)_LGT}

Our result applies to a wide range of physically relevant quantum systems, including the Hubbard-Holstein model and the U(1) and SU(2) LGTs. In \cite[Section~I]{TongAlbertMcCleanPreskillSu2021provably} it is discussed how they all satisfy the Assumption \ref{assumption:local_quantum_number_growth}. For concreteness, we will discuss in detail how the U(1) LGT fits into the framework of our current work.

In the 1D U(1) LGT, the system consists of a chain of $N$ nodes with $N-1$ links between adjacent nodes. We denote each node by $x$, and the links by the node on its left end. The links are sometimes called \textit{gauge links}.

On each node $x$ we have a fermionic mode whose annihilation operator is denoted by $\phi_x$. Each link consists of a planar rotor, whose configuration can be described by an angle $\theta\in[0,2\pi]$. 
The local Hilbert space is the vector space of square-integrable functions on $U(1)$.
An orthonormal basis of the local Hilbert space
can be chosen to be the Fourier basis (the electric basis).
More specifically we denote by $\ket{k}$ the Fourier mode $(2\pi)^{-1/2}e^{ik\theta}$, and $\{\ket{k}:k\in\ZZ\}$ form the basis we need.

We further define the operators $E_{x}$ and $U_{x}$, which act on the vector space of the links, through
$
E_{x}\ket{k} = k\ket{k}, U_{x}\ket{k} = \ket{k-1}.
$
The Hamiltonian for $U(1)$ LGT can then be described in terms of these operators via
\begin{equation}
    H = H_M +H_{GM} + H_E,
    \end{equation}
    where the three terms $H_M, H_{GM}, H_E$ describe the fermionic mass (using staggered fermions \cite{kogut1975hamiltonian}), the gauge-matter interaction, and the electric energy respectively,
\begin{equation}
\begin{aligned}
    H_M &= g_M\sum_{x}(-1)^x \phi_x^{\dagger}\phi_x, \\
    H_{GM} &= g_{GM}\sum_{x} (\phi_x^{\dagger}U_{x}\phi_{x+1}+\phi_{x+1}^{\dagger}U_{x}^{\dagger}\phi_{x}), \\
    H_E &= g_E\sum_{x} E_{x}^2.
\end{aligned}
\end{equation}

For LGTs, physical states need to satisfy Gauss's law: $G_x\ket{\Phi}=0$ for all physical states $\ket{\phi}$ where $G_x=E_x-E_{x-1}-\rho_x$, $\rho_x=\phi_x^{\dagger}\phi_x+((-1)^x-1)/2$. To ensure that the ground state of the Hamiltonian satisfy Gauss's law, we add a penalty term to the original Hamiltonian so that it becomes $H=H_M +H_{GM} + H_E+\lambda_G \sum_x G_x^2$, where we assume that $\lambda_G=\Or(1)$.

We can first write the Hamiltonian as a sum of local terms $H=H_1+H_2+\cdots+H_N$, where $H_x=g_M(-1)^x \phi_x^{\dagger}\phi_x+g_{GM} (\phi_x^{\dagger}U_{x}\phi_{x+1}+\phi_{x+1}^{\dagger}U_{x}^{\dagger}\phi_{x})+g_E E_{x}^2+\lambda_G G_{x+1}^2$. We consider the site $x$ to consist of both the node $x$ and the link $x$. The truncation is done through projecting in the electric basis: we define $\Pi^{(x)}_{[-\Lambda,\Lambda]}=\sum_{|k|\leq\Lambda}\ket{k}\bra{k}$. The truncated local Hilbert space is therefore spanned by $\{\ket{k}:|k|\leq\Lambda\}\otimes\{\ket{0},\ket{1}\}$, where $\ket{0},\ket{1}$ are the states of the fermionic mode on node $x$. We define the local quantum number $\lambda_x=E_x$, which is consistent with this truncation.

Now let us first check Assumption \ref{assumption:local_quantum_numbers}. The dimension of the truncated local Hilbert space is clearly $d(\Lambda)=2(2\Lambda+1)$. For $\mathcal{N}(\Lambda)=\|H_x\Pi^{(x)}_{[-\Lambda,\Lambda]}\|$, direct calculation shows $\mathcal{N}(\Lambda)=\Or(\Lambda^2)$. Verifying that $\braket{|\lambda_x|}=\Or(1)$ is a non-trivial task, which we perform in Section \ref{sec:mean_abs_local_quantum_number_bound} (for the SU(2) LGT and the Hubbard-Holstein model as well).

Next we check Assumption \ref{assumption:local_quantum_number_growth}. To decompose the Hamiltonian into $H^{(x)}_W$ and $H^{(x)}_R$, we observe that the only term that changes the local quantum number on site $x$ is $g_{GM}(\phi_x^{\dagger}U_{x}\phi_{x+1}+\phi_{x+1}^{\dagger}U_{x}^{\dagger}\phi_{x})$, which we define to be $H^{(x)}_W$. Because $\braket{k|U_x|k'}=0$ if $|k-k'|> 1$, \eqref{eq:m_conditions_ham_general_1} is satisfied. All the other terms are collected into $H^{(x)}_R$. Note that $H^{(x)}_R$ acts non-trivially on site $x$ (through the term $g_E E_x^2$), but it does not change the local quantum number. Therefore by this definition \eqref{eq:m_conditions_ham_general_3} is satisfied. It only remains to check \eqref{eq:m_conditions_ham_general_2}, which, because in this example $H^{(x)}_W$ is bounded even without truncation, is true by choosing $\chi=2|g_{GM}|$ and $r=0$. 

By checking the Assumptions \ref{assumption:local_quantum_numbers} and \ref{assumption:local_quantum_number_growth}, we can see that our main result (Theorem \ref{thm:area_law} from the supplemental material) applies to the 1D U(1) LGT, and therefore the area law is established provided that the spectral gap remains bounded away from $0$ as the system size increases. A similar procedure can be applied to the Hubbard-Holstein model and SU(2) LGT in 1D as well.

\sectionprl{Conclusion}
\label{sec:conclusion}

We have rigorously established an entanglement area law for the gapped ground state of 1D quantum systems with infinite-dimensional local Hilbert spaces under natural assumptions 
on the local quantum numbers. In particular, our entanglement area law applies to U(1) and SU(2) LGTs and the Hubbard-Holstein model in 1D, and the result may be adapted to handle other bosonic and gauge theory models of interest. 
Proving the entanglement area law is an important step toward justifying the use of tensor network methods in classical simulation of quantum systems. The proof techniques also provide many useful tools for designing a rigorous RG algorithm for these quantum systems, which should be investigated by future work. It is also worth considering how the result in this work can be generalized to degenerate ground state or low-energy states as in Ref.\ \cite{ALVV17}.

It is interesting that our proof of the area law uses seemingly different techniques from those of \cite{Hastings2007area}. This hints at a deeper underlying structure that empowers both proof techniques. We will leave exploration for what this could be for future work.


Our result relies on the local quantum number tail bound proved in Ref.\ \cite{TongAlbertMcCleanPreskillSu2021provably}, which in turn follows from technical tools for analyzing the dynamical simulation of unbounded Hamiltonians on digital quantum computers. Along this line, previous work such as\ \cite{anshu2020communication,kuwahara2021improved,nachtergaele2006lieb,Tran2019powerlaw} have found other applications of quantum simulation techniques, to solving problems in quantum many-body physics beyond quantum computing. We consider applications of this kind to be very interesting, 
 as they demonstrate important byproducts of the study of quantum algorithms which are of interest independently of the hope of building scalable quantum computers.
\sectionprl{Acknowledgments}

This work is partially supported by the Department of Energy under the Quantum System Accelerator (QSA) program (Y.T.) and the Spin Chain Bootstrap under BES (N.B.), by the National Science Foundation under the Quantum Leap Challenge Institute (QLCI) program through grant number OMA-2016245, NSERC Discovery Program (N.W.), Google Inc (N.W.) and by the Simons Foundation (N.A.). Y.S.\ thanks Andrew M.\ Childs, Abhinav Deshpande, Adam Ehrenberg, Alexey V.\ Gorshkov, Andrew Guo, and Minh Cong Tran for helpful discussions about the entanglement area law.  N.W.\ thanks Yuta Kikuchi for useful preliminary discussions about area laws for unbounded operators.

\bibliographystyle{abbrvnat}
\bibliography{combinedrefs}

\phantom{.}
\newpage
\phantom{.}
\newpage

\setcounter{page}{0}
\setcounter{secnumdepth}{2}
\onecolumngrid
\begin{center}
{\textsc{\Large Supplement to\\\large Entanglement area law for 1D gauge theories and bosonic systems}}
\end{center}

\medskip
\noindent \textit{Organization of the supplemental materials.} 
In Section \ref{sec:gauge_theories_and_bosonic_systems} we will discuss examples of quantum systems that are covered in the framework of this work. 
In the remaining parts of the supplemental materials we will provide the detailed proof of our results.
In Section \ref{sec:effective_dims} we will discuss properties of the ground state that can be used to truncate the infinite-dimensional local Hilbert space. In Section \ref{sec:mean_abs_local_quantum_number_bound} we prove that the local quantum numbers have a mean absolute value that can be bounded independently of the system size. 
This fact removes an otherwise extra assumption in the proof of the area law.
In Section \ref{sec:ham_truncation} we will demonstrate that the ground state is robust to such truncations using properties derived in the previous section. In Section \ref{sec:area_law} we prove our main result of the entanglement area law scaling, using an ``off-the-shelf'' lemma slightly tighter than the one from \cite{OTR}. 

\section{1D gauge theories and bosonic systems}
\label{sec:gauge_theories_and_bosonic_systems}

We now introduce some specific quantum systems that we will study and extract the common structure of these systems. These quantum systems are all discussed in \cite[Section~I]{TongAlbertMcCleanPreskillSu2021provably}.
We first consider the 1D Hubbard-Holstein model \cite{holstein1959studies}, 
a model describing electron-phonon interactions. 

\vspace{1em}
\textit{The Hubbard-Holstein model.} This model is defined on a 1D chain of $N$ nodes.
Each node in the lattice, indexed by $x$, contains two fermionic modes (spin up and down) and a bosonic mode. 
The Hamiltonian is
\begin{equation}
    \label{eq:ham_hubbard_holstein}
    H = H_f + H_{fb} + H_b,
\end{equation}
where $H_{f}$ is the Hamiltonian of the Fermi-Hubbard model \cite{Hubbard1963electron} acting on only the fermionic modes,
\begin{equation}
    H_{fb} =  g\sum_{x=1}^N (b_x^{\dagger}+b_x)(n_{x,\uparrow}+n_{x,\downarrow}-1), 
\end{equation}
is the boson-fermion coupling, and
\[
    H_b = \omega_0 \sum_{x=1}^N b_{x}^{\dagger}b_x,
\]
the purely bosonic parts of the Hamiltonian. Here, $b_x$ is the bosonic annihilation operator on node $x$, and $n_{x,\sigma}$ is the fermionic number operator for node $x$ and spin $\sigma$. 
\vspace{1em}

For gauge theories, we consider the Hamiltonian formulation of the U(1) and SU(2) LGTs \cite{kogut1975hamiltonian} in one dimension. The U(1) LGT is also known as the Schwinger model.

\vspace{1em}
\textit{The U(1) lattice gauge theory.}
The system consists of a chain of $N$ nodes with $N-1$ links between adjacent nodes. We denote each node by $x$, and the links by the node on its left end. The links are sometimes called \textit{gauge links}.

On each node $x$ we have a fermionic mode whose annihilation operator is denoted by $\phi_x$. Each link consists of a planar rotor, whose configuration can be described by an angle $\theta\in[0,2\pi]$. 
The local Hilbert space is the vector space of square-integrable functions on $U(1)$.
An orthonormal basis of the local Hilbert space
can be chosen to be the Fourier basis.
More specifically we denote by $\ket{k}$ the Fourier mode $(2\pi)^{-1/2}e^{ik\theta}$, and $\{\ket{k}:k\in\ZZ\}$ form the basis we need.

We further define the operators $E_{x}$ and $U_{x}$, which act on the vector space of the links, through
\begin{equation}
\label{eq:operators_U(1)}
    E_{x}\ket{k} = k\ket{k},\qquad U_{x}\ket{k} = \ket{k-1}.
\end{equation}
The Hamiltonian for $U(1)$ LGT can then be described in terms of these operators via
\begin{equation}
\label{eq:LGT_ham}
    H = H_M +H_{GM} + H_E,
    \end{equation}
    where the three terms $H_M, H_{GM}, H_E$ describe the fermionic mass (using staggered fermions \cite{kogut1975hamiltonian}), the gauge-matter interaction, and the electric energy respectively,
\begin{equation}
\begin{aligned}
    H_M &= g_M\sum_{x}(-1)^x \phi_x^{\dagger}\phi_x, \\
    H_{GM} &= g_{GM}\sum_{x} (\phi_x^{\dagger}U_{x}\phi_{x+1}+\phi_{x+1}^{\dagger}U_{x}^{\dagger}\phi_{x}), \\
    H_E &= g_E\sum_{x} E_{x}^2.
\end{aligned}
\end{equation}

In the context of the U(1) LGT, we need to ensure that the quantum state satisfies Gauss's law: 
\begin{equation}
    G_x\ket{\Phi}=0
\end{equation}
for all physical states $\ket{\phi}$ where 
\begin{equation}
    G_x=E_x-E_{x-1}-\rho_x,\quad \rho_x=\phi_x^{\dagger}\phi_x+((-1)^x-1)/2.
\end{equation} 
To ensure that the ground state of the Hamiltonian satisfy Gauss's law, we add a penalty term to the original Hamiltonian so that it becomes $H=H_M +H_{GM} + H_E+\lambda_G \sum_x G_x^2$, where we assume that $\lambda_G=\Or(1)$.

\vspace{1em}
\textit{The SU(2) lattice gauge theory.} 
We consider the theory using the fundamental representation of SU(2), as done in \cite{BanulsCichyEtAl2017efficient}. 
Each node $x$ now contains two fermionic modes, whose annihilation operators are denoted by $\phi_{x}^l$, $l=1,2$. 
Each gauge link consists of a rigid rotator whose configuration is described by an element of the group SU(2) \cite{kogut1975hamiltonian}. 

The Hamiltonian takes the form \eqref{eq:LGT_ham}, and is invariant under $\mathrm{SU}(2)$ transformations acting either from the left or from the right, which correspond to rotations of the rigid rotator with respect to space-fixed or body-fixed axes respectively. 

Physical states in SU(2) LGT also needs to satisfy Gauss's law, which takes a form that is similar to the U(1) case. We also ensure that it is satisfied in the ground state we study by adding a penalty term to the Hamiltonian.

\section{Truncating the local Hilbert space}
\label{sec:effective_dims}

Although we consider the setting where the local Hilbert spaces are infinite dimensional, we can approximate spectrally isolated eigenstates with states containing low local quantum numbers.
This fact is made rigorous in \cite[Theorem 12]{TongAlbertMcCleanPreskillSu2021provably} which we restate here in a slightly modified way:
\begin{thm}[Quantum number distribution tail bound]
\label{thm:tail_bound}
    Let $H$ be a Hamiltonian satisfying~Assumption \ref{assumption:local_quantum_number_growth}.
    Let $\ket{\Psi}$ be an eigenstate of $H$ corresponding to an eigenvalue $\varepsilon$ with multiplicity $1$, and $\varepsilon$ be separated from the rest of the spectrum of $H$ by a spectral gap $\Delta$. Moreover, for a fixed site $x$ we assume local quantum number $|\lambda_x|$ has a finite expectation value:
    \[
    \sum_{\lambda}|\lambda|\braket{\Psi|\Pi^{(x)}_{\lambda}|\Psi}=\bar{\lambda}_x<\infty,
    \]
    we then have that
    \[
    \|(I-{\Pi}^{(x)}_{[-\Lambda,\Lambda]})\ket{\Psi}\| \leq e^{-\Omega\left(\sqrt{\chi^{-1}\Delta(\Lambda^{1-r}-(2\bar{\lambda}_x)^{1-r})}\right)}.
    \]
\end{thm}

We will further need to apply this result to bound the error that emerges from applying the truncation bound to many sites.  The following corollary provides such a result in a convenient form.
\begin{cor}
\label{cor:multiple_truncation}
Let 
$
\Pi' = \prod_{x=\ell+1}^{\ell+s}\Pi^{(x)}_{[-\Lambda,\Lambda]},
$
then under the same assumption as in Theorem~\ref{thm:tail_bound}, we have
\begin{equation*}
    \|(I-{\Pi}')\ket{\Psi}\| \leq \sqrt{s} e^{-\Omega\left(\sqrt{\chi^{-1}\Delta(\Lambda^{1-r}-(2\bar{\lambda})^{1-r})}\right)},
\end{equation*}
where $\bar{\lambda}=\max_{\ell< x\leq \ell+s} \bar{\lambda}_x$.
\end{cor}
\begin{proof}
Proof follows from~Theorem~\ref{thm:tail_bound} via a straightforward application of the triangle inequality and the sub-multiplicative property of the spectral norm:
\begin{align}
    \|(I-{\Pi}')\ket{\Psi}\| &= \sqrt{\bra{\Psi}(I-{\Pi}')^2\ket{\Psi}}=\sqrt{\bra{\Psi}(I-{\Pi}')\ket{\Psi}}\nonumber\\
    &=\sqrt{ \sum_{x=\ell+1}^{\ell +s} \bra{\Psi}\prod_{x'=\ell+1}^{x-1}\Pi_{[-\Lambda,\Lambda]}^{(x')}(I - \Pi^{(x)}_{[-\Lambda,\Lambda]})\ket{\Psi}}\nonumber\\
    &\le\sqrt{ \sum_{x=\ell+1}^{\ell +s} \|(I - \Pi^{(x)}_{[-\Lambda,\Lambda]})\ket{\Psi}\|}\le \sqrt{s} e^{-\Omega\left(\sqrt{\chi^{-1}\Delta(\Lambda^{1-r}-(2\bar{\lambda})^{1-r})}\right)}. 
\end{align}
\end{proof}
A final important consequence of Theorem~\ref{thm:tail_bound} is a bound on the quantity $\|\Pi'H(I-{\Pi}')\ket{\Psi}\|$ where $\Pi'$ is defined as in Corollary~\ref{cor:multiple_truncation}.  We use the above corollary to derive such a bound below.
\begin{cor}
\label{cor:truncation_and_ham_error}
Let 
$
\Pi' = \prod_{x=\ell+1}^{\ell+s}\Pi^{(x)}_{[-\Lambda,\Lambda]},
$
then under the same assumption as in Theorem~\ref{thm:tail_bound}, we have
\begin{equation*}
    \|\Pi' H(I-{\Pi}')\ket{\Psi}\| \leq s^{3/2} \mathcal{N}(\Lambda) e^{-\Omega\left(\sqrt{\chi^{-1}\Delta(\Lambda^{1-r}-(2\bar{\lambda})^{1-r})}\right)},
\end{equation*}
where $\bar{\lambda}=\max_{\ell< x\leq \ell+s} \bar{\lambda}_x$.
\end{cor}

\begin{proof}
We only need to use Corollary~\ref{cor:multiple_truncation} along with a bound for $\|\Pi' H(I-{\Pi}')\|$. We have
\begin{equation}
    \Pi' H(I-{\Pi}') = \sum_{x=1}^N \Pi' H_x (I-{\Pi}') = \sum_{x=\ell+1}^{\ell+s} \Pi' H_x (I-{\Pi}').
\end{equation}
The second equality is because $\Pi' H_x (I-{\Pi}')=\Pi' (I-{\Pi}') H_x =0$ if $x\notin\{\ell+1,\ldots,\ell+s\}$. Therefore
\begin{equation}
    \|\Pi' H(I-{\Pi}')\| \leq \sum_{x=\ell+1}^{\ell+s} \|\Pi' H_x (I-{\Pi}')\|\leq \sum_{x=\ell+1}^{\ell+s} \|\Pi^{(x)}_{[-\Lambda,\Lambda]} H_x\| \leq s\mathcal{N}(\Lambda).
\end{equation}
\end{proof}

We remark that the bound in the above corollary does not depend on the system size $N$. However, there is a dependence on $\bar{\lambda}$, an upper bound on the mean absolute value of the local quantum numbers on sites $\ell+1,\ell+2,\ldots,\ell+s$. 
One might worry that $\bar{\lambda}$ will show up as an independent parameter in our expression of the entanglement entropy.
However, we will show in the next section that for the ground states of U(1) and SU(2) LGTs as well as the Hubbard-Holstein model, $\bar{\lambda}$ depends only on the coefficients in the Hamiltonian, and is thus independent of the system size.

\section{Bounding the mean absolute value of the local quantum number}
\label{sec:mean_abs_local_quantum_number_bound}

In this section we prove bounds on the mean absolute value $\bar{\lambda}_x=\braket{|\lambda_x|}$ of the local quantum numbers in the ground states of U(1) and SU(2) LGTs, as well as the Hubbard-Holstein model. $\lambda_x$ here is the local quantum number of site $x$ and $\braket{\cdot}$ denotes the ground state expectation value. We will drop the subscript $x$ in this section because we will focus on only a single bosonic mode or gauge link. These bounds are independent of the system size and only depend on the coefficients in the Hamiltonians. 

In the following discussion, we view our lattice models as a bipartite system, where a subsystem $A$ is a gauge link or bosonic mode, and a subsystem $B$ is the rest of the system. We then bound $\braket{|\lambda|}$, where $\lambda$ is the local quantum number for $A$, using the variational principle.

\begin{lem}
\label{lem:local_quantum_number_bound_first_step}
Consider a bipartite system with Hilbert space $\mathcal{H}=\mathcal{H}_A\otimes\mathcal{H}_B$, where $\mathcal{H}_A$ and $\mathcal{H}_B$ are the Hilbert spaces for subsystems $A$ and $B$ respectively.
Let $H=H_A+H_{AB}+H_B$ be the Hamiltonian, where $H_A$ acts non-trivially only on $A$, and $H_B$ on $B$. Let $\ket{\Psi}$ be the ground state of $H$. Furthermore, we assume that there exists an operator $K_A$ acting non-trivially only on $A$ such that $|H_{AB}|\preceq K_A$. Then we have
\begin{equation}
\label{eq:local_expectation_value_bound}
    \braket{\Psi|\left(H_A-K_A\right)|\Psi}\leq\braket{\Psi_A|\left(H_A+K_A\right)|\Psi_A},
\end{equation}
for any $\ket{\Psi_A}\in\mathcal{H}_A$.
Here $\preceq$ denotes the partial order induced by the convex cone of positive semi-definite operators and $|O|=\sqrt{O^{\dagger}O}$ for operator $O$.
\end{lem}

\begin{proof}

Let $\ket{\Psi_P}=\ket{\Psi_A}\ket{\Psi_B}$ be a product state, where $\ket{\Psi_B}\in\mathcal{H}_B$ is the ground state of $H_B$, and $\ket{\Psi_A}\in\mathcal{H}_A$ is chosen arbitrarily. Then because $\ket{\Psi}$ is the ground state of $H$ we have
\begin{equation}
\label{eq:variational_principle}
    \braket{\Psi|H|\Psi}\leq \braket{\Psi_P|H|\Psi_P}.
\end{equation}

For $\braket{\Psi|H|\Psi}$ we have
\begin{equation}
\label{eq:variational_principle_LHS}
    \begin{aligned}
    \braket{\Psi|H|\Psi} &= \braket{\Psi|H_A|\Psi} + \braket{\Psi|H_{AB}|\Psi} + \braket{\Psi|H_B|\Psi} \\
    &\geq \braket{\Psi|H_A|\Psi} - \braket{\Psi|K_A|\Psi} + \braket{\Psi|H_B|\Psi} \\
    &\geq \braket{\Psi|\left(H_A-K_A\right)|\Psi} + \braket{\Psi_B|H_B|\Psi_B},
    \end{aligned}
\end{equation}
where in the first inequality we have used $|H_{AB}|\preceq K_A$ and in the second inequality we have used  $\braket{\Psi|H_B|\Psi}\geq \braket{\Psi_B|H_B|\Psi_B}$, which is true because $\ket{\Psi_B}$ is chosen to be the ground state of $H_B$.

For $\braket{\Psi_P|H|\Psi_P}$ we have
\begin{equation}
\label{eq:variational_principle_RHS}
    \begin{aligned}
    \braket{\Psi_P|H|\Psi_P} &= \braket{\Psi_P|H_A|\Psi_P} + \braket{\Psi_P|H_{AB}|\Psi_P} + \braket{\Psi_B|H_B|\Psi_B} \\ 
    &\leq \braket{\Psi_P|\left(H_A+K_A\right)|\Psi_P} + \braket{\Psi_B|H_B|\Psi_B}  \\
    &= \braket{\Psi_A|\left(H_A+K_A\right)|\Psi_A} + \braket{\Psi_B|H_B|\Psi_B}.
    \end{aligned}
\end{equation}

Combining \eqref{eq:variational_principle}, \eqref{eq:variational_principle_LHS}, \eqref{eq:variational_principle_RHS}, we have \eqref{eq:local_expectation_value_bound}.

\end{proof}

\begin{lem}
\label{lem:local_quantum_number_bound_general}
Under the same assumptions as in Lemma \ref{lem:local_quantum_number_bound_first_step}, we further let $\Xi$ be a Hermitian operator on $A$. Assume that $H_A-K_A\succeq L(|\Xi|)$ where $L(x)$ is a convex non-decreasing function for $x\in\RR^+$ satisfying $L(x)\to+\infty$ when $x\to+\infty$. Then we have
\begin{equation}
\label{eq:mean_abs_bound_general}
    \braket{\Psi||\Xi||\Psi}
    \leq L^{-1}\left(\braket{\Psi|\left(H_A-K_A\right)|\Psi}\right)
    \leq L^{-1}\left(\min_{\ket{\Psi_A}\in\mathcal{H}_A}\braket{\Psi_A|\left(H_A+K_A\right)|\Psi_A}\right).
\end{equation}
\end{lem}

\begin{proof}
By the assumption that $H_A-K_A\succeq L(|\Xi|)$ and \eqref{eq:local_expectation_value_bound}, we have
\[
\braket{\Psi|L(|\Xi|)|\Psi} \leq \braket{\Psi_A|\left(H_A+K_A\right)|\Psi_A}.
\]
Because $L$ is convex, by Jensen's inequality $L(\braket{\Psi||\Xi||\Psi})\leq\braket{\Psi|L(|\Xi|)|\Psi}$. Because $L$ is non-decreasing and $\ket{\Psi_A}$ can be arbitrarily chosen, we have \eqref{eq:mean_abs_bound_general}.
\end{proof}

Note that the right-hand side of \eqref{eq:mean_abs_bound_general} is independent of the subsystem $B$. This allows us to bound the mean absolute value of the local quantum number in the lattice models in a way that is independent of the system size.
We will now apply this lemma to the case of gauge theories and the Hubbard-Holstein model. 

First, for U(1) and SU(2) gauge theories with the Hamiltonian defined in \eqref{eq:LGT_ham}, we take $A$ to be a gauge link $x$. Then we let
\[
    H_A = g_E E_x^2,\qquad H_{AB} = g_{GM}(\phi_x^{\dagger}U_x\phi_{x+1}+\phi_{x+1}^{\dagger}U_x^{\dagger}\phi_{x}),\qquad H_{B}=H-H_A-H_B.
\]
A nice feature about this Hamiltonian is that $H_{AB}$ is bounded: $\|H_{AB}\|\leq 2|g_{GM}|$. Therefore we can simply choose $K_A=2|g_{GM}|$. The local quantum number $\lambda_x$ in this case is the electric field value in the U(1) case and the total angular momentum in the SU(2) case. But $E_x^2 = \lambda_x^2$ for U(1) LGT and $E_x^2 = \lambda_x(\lambda_x+1)\succeq \lambda_x^2$ for SU(2) LGT (for the latter see section \ref{sec:opSU2}\comment{\eqref{eq:operators_SU(2)}} below. Also $\lambda_x\succeq 0$ for SU(2) LGT). Therefore we have
\[
H_A - K_A \succeq g_E \lambda_x^2 - 2|g_{GM}|
\]
for both cases. For the right-hand side of \eqref{eq:mean_abs_bound_general} we have
\[
\min_{\ket{\Psi_A}\in\mathcal{H}_A}\braket{\Psi_A|\left(H_A+K_A\right)|\Psi_A} = 2|g_{GM}|,
\]
where the minimum is attained by $\ket{\Psi_A}=\ket{0}$.
Combining the above facts, we apply Lemma~\ref{lem:local_quantum_number_bound_general} to get:
\begin{cor}
\label{cor:mean_abs_bound_LGT}
For the U(1) and SU(2) LGTs with the Hamiltonian defined in \eqref{eq:LGT_ham}, let $\lambda_x$ be the local quantum number on gauge link $x$, then $\braket{|\lambda_x|}\leq 2\sqrt{|g_{GM}|/g_E}$, where $\braket{\cdot}$ denotes the ground state expectation value.
\end{cor}
 Then let us consider the Hubbard-Holstein model with the Hamiltonian described in \eqref{eq:ham_hubbard_holstein}. We let $A$ be a bosonic mode $x$, and let $B$ be the rest of the system. We have
\[
H_A = \omega_0 b_x^{\dagger}b_x,\qquad H_{AB}=g(b_x^{\dagger}+b_x)(n_{x,\uparrow}+n_{x,\downarrow}-1),\qquad H_B = H-H_A-H_{AB}.
\]
Here $H_{AB}$ is no longer bounded, but we can still construct a $K_A$ such that $|H_{AB}|\preceq K_A$.
To simplify the discussion we introduce the position and momentum operators $X$ and $P$:
\[
X = \frac{1}{\sqrt{2}}(b_x^{\dagger}+b_x),\qquad P=\frac{i}{\sqrt{2}}(b_{x}^{\dagger}-b_x).
\]
Then $H_{AB}=\sqrt{2}g X(n_{x,\uparrow}+n_{x,\downarrow}-1)$. Because $\|n_{x,\uparrow}+n_{x,\downarrow}-1\|\leq 1$, we can define
\[
K_A = \sqrt{2}|g| |X|,
\]
which satisfies $|H_{AB}|\preceq K_A$. With this choice of $K_A$ we have
\[
\begin{aligned}
H_A - K_A &= \frac{\omega_0}{2}(X^2+P^2-1)-\sqrt{2}|g||X| \\
&\succeq \frac{\omega_0}{2}(X^2+P^2-1) - \frac{\omega_0}{4}X^2 - \frac{2g^2}{\omega_0} \\
&\succeq \frac{\omega_0}{4}(X^2+P^2-1) - \frac{\omega_0}{4}  - \frac{2g^2}{\omega_0} \\
&= \frac{\omega_0}{2}b_x^{\dagger}b_x - \frac{\omega_0}{4} - \frac{2g^2}{\omega_0}.
\end{aligned}
\]
The local quantum number here is the bosonic occupation number, which has to be non-negative. Therefore
\[
H_A - K_A \succeq \frac{\omega_0}{2}|\lambda_x| - \frac{\omega_0}{4} - \frac{2g^2}{\omega_0}.
\]
For the right-hand side of \eqref{eq:mean_abs_bound_general}, we have
\[
\min_{\ket{\Psi_A}\in\mathcal{H}_A}\braket{\Psi_A|\left(H_A+K_A\right)|\Psi_A}\leq \braket{0|\left(\omega_0b_x^{\dagger}b_x + 2|g||X|\right)|0} = \frac{2|g|}{\sqrt{\pi}},
\]
where we have used the analytic solution of the ground state of the harmonic oscillator in deriving the equality. Combining these results with Lemma~\ref{lem:local_quantum_number_bound_general} we have:
\begin{cor}
\label{cor:mean_abs_bound_Hubbard_Holstein}
For the Hubbard-Holstein model with the Hamiltonian defined in \eqref{eq:ham_hubbard_holstein}, let $\lambda_x$ be the local quantum number on site $x$, then 
\[
\braket{|\lambda_x|}\leq \frac{1}{2}+\frac{4|g|}{\omega_0 \sqrt{\pi}} + \frac{4g^2}{\omega_0^2},
\]
where $\braket{\cdot}$ denotes the ground state expectation value.
\end{cor}

\subsection{SU(2) LGT}
\label{sec:opSU2}
For the SU(2) case the operators $E^2_{x}$ and $U_{x}$ are different from the U(1) case.
The operator $E^2_{x}$ is defined through
\begin{equation}
\label{eq:operators_SU(2)}
    E^2_{x}\ket{jmm'} = j(j+1)\ket{jmm'}.
\end{equation}
Because $\phi_x$ has two components, where each component is a fermionic mode,  $U_{x}$ is a $2\times 2$ matrix, where each of the 4 matrix entries is an operator acting on the link space
\begin{equation}
    U_{x} = \begin{pmatrix}
    U^{11}_{x} & U^{12}_{x}\\
    U^{21}_{x} & U^{22}_{x}
    \end{pmatrix}.
\end{equation}
Given that $U_{x,n_i}$ transforms as the $j=1/2$ representation of $\mathrm{SU}(2)$, the rules for the addition of angular momentum imply
\begin{equation}
    \label{eq:U_SU(2)_properties}
    \begin{aligned}
    \braket{j_1 m_1 m_1'|U^{ll'}_{x}|j_2 m_2 m_2'}&=0,\ \text{if }|j_1-j_2|>1/2,
    \\\|U^{ll'}_{x}\|&\leq 1.
    \end{aligned}
\end{equation}
Here $\|O\|$ denotes the spectral norm of an operator $O$.

\section{Robustness of the ground state to truncation}
\label{sec:ham_truncation}

In this section, we show that the ground state, the ground state energy, and the spectral gap are all robust to the truncation of the Hamiltonian, in a way that we will specify later.
Following \cite{AradKitaevLandauVazirani2013area} we focus on the $s$ sites from $\ell+1$ to $\ell+s$. For convenience we relabel the sites so that the original site $x$ is now labelled $x-\ell$. The Hamiltonian can be rewritten as
\begin{equation}
\label{eq:local_ham_general}
    H = H_L + H_1 + \cdots + H_s + H_R,
\end{equation}
where $H_L = \sum_{x\leq 0} H_x$ and $H_R = \sum_{x\geq s+1} H_x$. We need to shift each Hamiltonian term ($H_L$, $H_1$, $H_2,\ldots,H_s$, $H_R$) by a constant to ensure that they are all positive semi-definite. As in \cite{AradKitaevLandauVazirani2013area} we look at the entanglement entropy across a cut between sites $s/2$ and $s/2+1$. If $\|H_x\|\leq 1$, the local Hilbert space dimension is $d$ and the spectral gap is $\Delta$, then it is known that the entanglement entropy scales as $\Or(\log^3(d)/\Delta)$ \cite[Theorem 6.2]{AradKitaevLandauVazirani2013area}.

Now we want to consider the case where the local Hilbert space dimension is infinite. This compels us to truncate the local Hilbert space, and consequently the local Hamiltonian terms $H_i$ as well. We denote the truncation threshold, defined according to the local quantum number introduced in Section~\ref{sec:the_abstract_model}, by $\Lambda$, and correspondingly the truncated Hamiltonian term by $H_x'$. The truncated Hilbert space dimension is $d(\Lambda)$ and the truncated local term has a norm that is upper bounded by $\mathcal{N}(\Lambda)$. 

\subsection{The two truncations}

We first clarify  in more detail what we mean by truncating the local Hilbert space. The original Hilbert space is $\mathcal{H}=\mathcal{H}_L\otimes \mathcal{H}_1 \otimes \cdots \mathcal{H}_s\otimes \mathcal{H}_R$. We consider a subspace $\mathcal{H}'=\mathcal{H}_L\otimes \mathcal{H}'_1 \otimes \cdots \mathcal{H}'_s\otimes \mathcal{H}_R \subseteq \mathcal{H}$, where each $\mathcal{H}_x$ has dimension $d(\Lambda)$. 
We denote by $\Pi'_x$ the projection operator onto $\mathcal{H}'_x$, and define
\[
\Pi' = I_L \otimes \Pi'_1 \otimes \cdots \Pi'_s \otimes I_R,
\]
which is the projection operator onto $\mathcal{H}'$. $I_L$ and $I_R$ are the identity operators on $\mathcal{H}_L$ and $\mathcal{H}_R$ respectively.
The truncated Hamiltonian is a Hermitian operator $H'$ defined to be the \textit{restriction} of $\Pi'H\Pi'$ to the subspace $\mathcal{H}'$, by which we mean that $H'$ maps elements from $\mathcal{H}'$ to $\mathcal{H}'$. We do not directly define $H'$ to be $\Pi' H\Pi'$ because that would introduce an artificial eigenvalue $0$ corresponding to the part of the Hilbert space that is truncated out.
We can write $H'$ out as
\begin{equation}
    H'=H_L' + H_1' + \cdots + H_s' + H_R', 
\end{equation}
where each $H'_x$ is the restriction of $\Pi' H_x\Pi'$ to the subspace $\mathcal{H}'$, and the same is true for $H'_L$ and $H'_R$.  Locality is preserved in this truncation as $H'_x$ still acts non-trivially on sites $x$ and $x+1$, $H'_L$ on sites to the left of and including site $1$, and $H'_R$ on sites to the right of site $s$.

The second truncation we consider comes from Ref.~\cite{AradKitaevLandauVazirani2013area}. We adopt the definition in \cite{AradKitaevLandauVazirani2013area}, for a Hermitian operator $A$,
\begin{equation}
A^{\leq t} = AP_t + \|AP_t\|(I-P_t),
\end{equation}
where $P_t$ is the projection operator onto the subspace spanned by eigenvectors of $A$ with eigenvalues at most $t$. Then the second truncation yields the Hamiltonian
\begin{equation}\label{eq:Hpp}
H'' = (H_L' + H_1')^{\leq t} + H_2' + \cdots + H_{s-1}' + (H_s' + H_R')^{\leq t}.
\end{equation}
 The goal here is to show that these two truncations (i) preserve the spectral gap up to a constant factor, and (ii) preserve the ground state up to an error exponentially small in $\Lambda$ and $t$.

\subsection{Truncation robustness of the ground state and energy}

\begin{defn}\label{epsdef}For a self-adjoint operator $A$ bounded from below, define the sequence $\epsilon_0(A)\le\epsilon_1(A)\le\ldots$ as follows. Let $\essspec(A)$ be its (closed) essential spectrum and let $K\in\mathbb{N}_0\cup\{\infty\}$ be the number of eigenvalues in $[-\infty,\min\essspec(A))$, including multiplicity. For $k<K$, let $\epsilon_k(A)$ be the $k$-th eigenvalue with multiplicity ($K$ eigenvalues as we start from $k=0$), and for $k\ge K$ let $\epsilon_k(A)=\min\essspec(A)$.\end{defn}
Each $\epsilon_k(A)$ is in the spectrum\footnote{I.e., the set of $\epsilon$ such that $A-\epsilon I$ has no bounded inverse.} of $A$ but is not necessarily an eigenvalue when $\epsilon_k(A)=\min\essspec(A)$. The min-max principle \cite[Theorem 4.10]{te2009mathematical} states that for $k\in\mathbb N_0$, 
 \begin{equation}\label{minmaxthm}
 \epsilon_k(A)=\minmax{\in\mathcal{D}(A)}{k}{A},
 \end{equation}
 where $\mathcal{D}$ is the domain of $A$.

Writing $\epsilon_k=\epsilon_k(H)$, recall that we assume that the Hamiltonian $H$ has a non-degenerate lowest eigenvalue $\epsilon_0$ (the ground state energy), with a unique ground state $\ket{\psi_0}$. We also assume that $\epsilon_0$ is separated from the rest of the spectrum by a gap $\Delta=\epsilon_1-\epsilon_0$.

For truncated Hamiltonians $H'$ and $H''$, we will prove that there exist unique ground states $\ket{\psi_0'}$ and $\ket{\psi_0''}$, corresponding to non-degenerate lowest eigenvalues $\epsilon_0'$ and $\epsilon_0''$, for the two truncated Hamiltonians respectively. Write $\epsilon_k=\epsilon_k(H)$, $\epsilon_k'=\epsilon_k(H')$, and $\epsilon_k''=\epsilon_k(H'')$ and $\Delta'=\epsilon_1'-\epsilon_0'$, $\Delta''=\epsilon_1''-\epsilon_0''$.

\begin{thm}[Robustness to truncations]
\label{thm:robustness}
Let $\Pi'$ be the projection operator onto $\mathcal{H}'$.
Let $\delta_1 = \|(I-\Pi')\ket{\psi_0}\|$, $\delta_2=\|\Pi'H(I-\Pi')\ket{\psi_0}\|$, and 
\begin{equation}
\frac{\delta_2}{1-\delta_1^2} \leq \frac{\Delta}{18},
\end{equation}
then for every cutoff of the local quantum number $\Lambda>0$ there exists 
\begin{equation}
T=\Or\left(\frac{\mathcal{N}(\Lambda)^2}{\Delta}\left(\frac{\epsilon_0+\delta_2/(1-\delta_1^2)}{\Delta}+1\right)\right),
\end{equation}
such that for all $t\geq T$,
\begin{itemize}
\item[(i)] For $H''$, there exists a non-degenerate ground state $\ket{\psi_0''}$ corresponding to the lowest eigenvalue $\epsilon_0''$.
\item[(ii)] $\Delta''=\Omega(\Delta)$;
\item[(iii)] The trace distance between $\ket{\psi_0}$ and $\ket{\psi''_0}$ can be bounded as follows:
\begin{equation}
D(\ket{\psi_0''},\ket{\psi_0})\le \sqrt{\frac{2\delta_2}{\Delta(1-\delta_1^2)}} + e^{-\Omega(t/\mathcal{N}(\Lambda))}.
\end{equation}
\end{itemize}
\end{thm}

We note that in the above theorem we require the eigenvalue cutoff $t$ for $H_L$ and $H_R$ to be above a certain $T$, due to a similar requirement in \cite[Theorem 6.1]{AradKitaevLandauVazirani2013area}. The scaling of $T$ has later been improved in Ref. \cite{Huang2014area}.

Before we proceed with the proof we establish the following lemma, which follows from a similar reasoning as in \cite[Lemma 6.4]{AradKitaevLandauVazirani2013area} except we correct a minor mistake in their bound due to not taking the global phase into account.
\begin{lem}[Markov]
\label{lem:markov}
Let $H$ be a Hamiltonian with the lowest eigenvalue $\epsilon_0$ and all other eigenvalues at least $\epsilon_1>\epsilon_0$, and assume $\ket{\psi_0}$ is its unique ground state. Given a quantum state $\ket{\phi}$ with expectation value $\braket{\phi|H|\phi}=E$, we have
\begin{equation}
    |\braket{\phi|\psi_0}|^2 \geq \frac{\epsilon_1-E}{\epsilon_1-\epsilon_0}.
\end{equation}
\end{lem}

\begin{proof}
Since the eigenstate $\ket{\psi_0}$ is the unique ground state, the expectation value of $H$ in any other eigenstate must be at least $\epsilon_1$ by assumption.  Using the fact that $\epsilon_1-\epsilon_0>0$, we have from Markov's inequality 
\begin{equation}
    E=\bra{\phi} H \ket{\phi} \ge \epsilon_0 |\braket{\phi|\psi_0}|^2 + \epsilon_1(1-|\braket{\phi|\psi_0}|^2).
\end{equation}
The result then follows by re-arranging the above expression.
\end{proof}

The proof of Theorem~\ref{thm:robustness} proceeds as follows: we first show that when we go from $H$ to $H'$, the spectral gap is preserved and the ground state is changed by a small amount, and then show the same is true when we go from $H'$ to $H''$. In the first step we obtain:

\begin{lem}
\label{lem:robustness_first_step}
Let $\Pi'$ be the projection operator onto $\mathcal{H}'$.
Let $\delta_1 = \|(I-\Pi')\ket{\psi_0}\|$, $\delta_2=\|\Pi'H(I-\Pi')\ket{\psi_0}\|$. Then if $\delta_1 < 1$ and
\begin{equation}
\label{eq:truncation_err_assumption_robust_lemma}
\frac{\delta_2}{1-\delta_1^2} \leq \frac{\Delta}{4},
\end{equation}
we have the following 
\begin{itemize}
\item[(i)] For $H'$, there exists a non-degenerate ground state $\ket{\psi_0'}$ corresponding to the lowest eigenvalue $\epsilon_0'$.
\item[(ii)] $\Delta'=\Omega(\Delta)$.
\item[(iii)] The trace distance between $\ket{\psi_0}$ and $\ket{\psi'_0}$ can be bounded as
\begin{equation}
D(\ket{\psi_0'},\ket{\psi_0})\leq\sqrt{\frac{2\delta_2}{\Delta(1-\delta_1^2)}}.
\end{equation}
\item[(iv)] $\epsilon_0\leq \epsilon'_0\leq \epsilon_0 + \frac{2\delta_2}{1-\delta_1^2}$.
\end{itemize}
Here $D(\cdot,\cdot)$ denotes the trace distance.
\end{lem}
\begin{proof}
  By the min-max theorem (Equation \ref{minmaxthm}), we have for $k\in\mathbb N_0$
 \begin{equation}\begin{aligned}
 \epsilon_k=&\minmax{\in\mathcal{H}}{k}{H}\\\le&\minmax{\in\mathcal{H'}}{k}{H}=\epsilon_k'\le\min\essspec(H').
 \end{aligned}\label{minmax}\end{equation}
In particular we have $\epsilon_1\le\epsilon_1'\le\min\essspec(H')$. To establish a gap $\Delta'$ we need an upper bound on $\epsilon_0'$:
 \begin{equation}
 \begin{aligned}
      \epsilon_0 &= \braket{\psi_0|H|\psi_0} \\
      &= \braket{\psi_0|\Pi' H\Pi'|\psi_0} + \braket{\psi_0|(I-\Pi') H\Pi'|\psi_0} \\
      &\quad+ \braket{\psi_0|\Pi' H(I-\Pi')|\psi_0} + \braket{\psi_0|(I-\Pi') H(I-\Pi')|\psi_0} \\
      &\geq \epsilon'_0 \|\Pi'\ket{\psi_0}\|^2 - 2\|\Pi' H(I-\Pi')\ket{\psi_0}\| + \epsilon_0 \|(I-\Pi')\ket{\psi_0}\|^2 \\
      &= \epsilon'_0 (1-\delta_1^2) -2\delta_2 + \epsilon_0 \delta_1^2.
 \end{aligned}
 \end{equation}
 As a result,
 \begin{equation}
 \label{eq:new_ground_energy_bound}
     \epsilon'_0 \le\epsilon_0+ \frac{2\delta_2}{1-\delta_1^2}.
 \end{equation}
The bound $\epsilon_0\le\epsilon_0'$ is immediate from the variational principle (or \eqref{minmax}), so (iv) is established. The assumption \eqref{eq:truncation_err_assumption_robust_lemma} then yields \[\epsilon_0'\le\epsilon_0+\Delta/2\le\epsilon_1-\Delta/2\le\epsilon_1'-\Delta/2,\]
 which implies $\Delta'\ge\Delta/2$, hence (ii). In particular it implies (i), that $\epsilon_0'$ is a simple eigenvalue with eigenvector $\ket{\psi_0'}$.
 
To establish closeness between $\ket{\psi_0'}$ and $\ket{\psi_0}$ we apply Lemma~\ref{lem:markov},
 \[\braket{\psi_0'|\psi_0}^2\ge\frac{\epsilon_1-\braket{\psi_0'|H|\psi_0'}}{\epsilon_1-\epsilon_0}=\frac{\epsilon_1-\epsilon_0'}\Delta=1-\frac{\epsilon_0'-\epsilon_0}{\Delta}\ge1-\frac1\Delta\frac{2\delta_2}{1-\delta_1^2},\]
 where the last inequality follows from \eqref{eq:new_ground_energy_bound}. Claim (iii) follows since $D(\ket{\psi'_1},\ket{\psi_0}) \leq \sqrt{1 - |\braket{\psi'_1|\psi_0}|^2}$.
 \end{proof}

We are then ready to prove the main result.
\begin{proof}[Proof of Theorem~\ref{thm:robustness}]
The existence and uniqueness of the lowest eigenvalue $\epsilon_0''$ and the ground state $\ket{\psi_0''}$ in (i) can be proved similarly to the proof of Lemma \ref{lem:robustness_first_step} (i).
We rescale the Hamiltonian $H'$ by a factor $\mathcal{N}(\Lambda)$ and then apply \cite[Theorem 6.1]{AradKitaevLandauVazirani2013area}. That theorem tells us that $\Delta''=\Omega(\Delta')$. Combining this fact with Lemma~\ref{lem:robustness_first_step} (ii), we have (ii). \cite[Theorem 6.1]{AradKitaevLandauVazirani2013area} also tells us that there exists 
\begin{equation}
    T = \Or\left(\frac{\mathcal{N}(\Lambda)^2}{\Delta'}\left(\frac{\epsilon_0'}{\Delta'}+1\right)\right) = \Or\left(\frac{\mathcal{N}(\Lambda)^2}{\Delta}\left(\frac{\epsilon_0+\delta_2/(1-\delta_1^2)}{\Delta}+1\right)\right),
\end{equation}
such that $D(\ket{\psi_0'},\ket{\psi_0''})=e^{-\Omega(t/\mathcal{N}(\Lambda))}$ for $t\geq T$, where $\mathcal{N}(\Lambda)$ comes from the rescaling. Here we have used Lemma~\ref{lem:robustness_first_step} (iv). Combining with Lemma~\ref{lem:robustness_first_step} (iii) and the triangle inequality, we have proved (iii).
\end{proof}

Now we can use Corollaries~\ref{cor:multiple_truncation} and \ref{cor:truncation_and_ham_error} to bound $\delta_1$ and $\delta_2$ in Theorem~\ref{thm:robustness}, which leads to the following robustness result:
\begin{cor}[Robustness to truncations]
\label{cor:robustness}
Assume that the Hamiltonian satisfies Assumptions~\ref{assumption:local_quantum_numbers} and \ref{assumption:local_quantum_number_growth}, and let 
$
    \bar{\lambda} = \max_{1\leq x\leq s} \braket{|\lambda_x|}.
$
Then the truncated Hamiltonian $H''$ has a lowest eigenvalue $\epsilon_0''$ corresponding to a non-degenerate ground state $\ket{\psi_0''}$. And there exist constants $C_1$ and $C_2$ such that for any $\Lambda$ and $t$ satisfying
\begin{equation}\label{eq:minparams}
\Lambda^{1-r}\geq (2\bar{\lambda})^{1-r} + C_1 \Delta^{-1}\polylog(s,\Delta^{-1}),\qquad t\geq \frac{C_2 \mathcal{N}(\Lambda)^2}{\Delta^2},
\end{equation}
we have
\begin{itemize}
\item[(i)] $\Delta''=\Omega(\Delta)$;
\item[(ii)] The trace distance between $\ket{\psi_0}$ and $\ket{\psi''_0}$ can be bounded as follows:
\begin{equation}
D(\ket{\psi_0''},\ket{\psi_0})=\poly(s,\Delta^{-1},\Lambda)e^{-\Omega\left(\sqrt{\Delta(\Lambda^{1-r}-(2\bar{\lambda})^{1-r})}\right)} + e^{-\Omega(t/\mathcal{N}(\Lambda))}.
\end{equation}
\end{itemize}
\end{cor}

We recall that $r$ is the exponent involved in \eqref{eq:m_conditions_ham_general_2}, and that $r=0$ for LGTs, $r=1/2$ for the Hubbard-Holstein model. $\Delta$ here is the spectral gap, and $s$ is the size of the region around the cut that we want to pay special attention to in \eqref{eq:local_ham_general}.

If we assume $\bar\lambda=O(1)$, which we proved for the Hubbard-Holstien model and U(1) and SU(2) LGTs in Section \ref{sec:mean_abs_local_quantum_number_bound}, then to achieve $D(\ket{\psi_0''},\ket{\psi})^2\le\delta$, it suffices to require, for some constant $C$,
\[\sqrt{\Delta(\Lambda^{1-r}-C)}\ge C\log\frac{s\Lambda}{\delta\Delta}\quad\text{and}\quad t\ge C\mathcal{N}(\Lambda)\log(1/\delta)\vee\frac{C\mathcal{N}(\Lambda)^2}{\Delta^2},\]
where $\vee$ denotes the maximum, which can be satisfied by choosing
\begin{equation}
\label{eq:Ltchoices}
\Lambda=\poly(\Delta^{-1})\polylog(s/\delta)\quad\text{and}\quad t=\Theta( \mathcal{N}(\Lambda)\log(1/\delta)\vee{\mathcal{N}(\Lambda)^2}/{\Delta^2}).
\end{equation}

\section{Area law}
\label{sec:area_law}
In this section, we establish our main result of an entanglement area law for unbounded quantum systems.
We first recall the notion of AGSP from \cite[Definition 2.1]{AradKitaevLandauVazirani2013area}.
\begin{defn}
\label{defn:agsp}
$K$ is a $(\sigma,R)$-AGSP of a Hamiltonian $H$ on a bipartite system consisting of two parts $A$ and $B$ that has a non-degenerate ground state, if
\begin{enumerate}
    \item $K\ket{\Psi}=\ket{\Psi}$, where $\ket{\Psi}$ is the ground state of $H$.
    \item $\|K\ket{\Phi}\|\leq \sigma$ for any $\ket{\Phi}$ such that $\braket{\Phi|\Psi}=0$.
    \item There exist operators $K_j^A$ and $K_j^B$, acting on $A$ and $B$ respectively, $j=1,2,\ldots,R$, such that $K=\sum_{j=1}^{R}K_j^A\otimes K_j^B$.
\end{enumerate}
\end{defn}

An AGSP preserves the ground state, suppresses the excited states, and increases the entanglement by a finite amount.
The existence of an AGSP onto the exact ground state is known to imply a bound on the entanglement of ground state. More precisely, Corollary III.4 of \cite{ALV12} states that if $\shrink R\le1/2$ where $\shrink$ is the shrinking factor and $R$ is the entanglement rank of the AGSP, then the entanglement entropy of the ground state satisfies a bound of order $\log R$. For frustrated systems the target space of the AGSP becomes perturbed away from the exact ground state(s) as its construction involves spectral truncations. Analyses of such a situation are undertaken in \cite{AradKitaevLandauVazirani2013area} and \cite{ALVV17}. These tools were simplified in \cite{OTR}, where an ``off-the-shelf'' lemma was stated which generalizes the one of \cite{ALV12} to perturbed and degenerate target spaces. 

Here we make a further improvement to \cite{OTR} to obtain the cleaner and slightly stronger statement of Lemma \ref{OTRlemma} below. For two subspaces $\mathcal Y,\mathcal Z\subset\mathcal H$, we say $\mathcal{Y}$ is $\delta$-viable for $\mathcal{Z}$ if $\braket{z|P_{\mathcal{Y}^\perp}|z}\leq\delta$ for all unit vectors $\ket{z}\in\mathcal{Z}$, where $P_{\mathcal Y^\perp}$ is the projection onto the orthogonal complement of subspace $\mathcal{Y}$. We write $\mathcal Y\approx_\delta\mathcal Z$ if $\bra z P_{\mathcal Y^\perp}\ket z\le\delta$ and $\bra y P_{\mathcal Z^\perp}\ket y\le\delta$ for unit vectors $\ket y\in\mathcal Y$ and $\ket z\in\mathcal Z$ respectively.

\begin{lem}\label{OTRlemma}
Let $\mathcal Z$ be a subspace of bipartite space $\mathcal H=\mathcal H_1\otimes\mathcal H_2$.
Let $\tilde{\mathcal Z}_n$ be a sequence of subspaces of $\mathcal H$ such that $\tilde{\mathcal Z}_n\approx_{\delta_n}\mathcal Z$ where $\delta_n$ is a sequence such that $\sum_{n=0}^\infty n\delta_n=O(1)$.

Suppose there exist $K_1,K_2,\ldots$ such that $K_n$ is an $(\shrink^n,R^n)$-AGSP with target space $\tilde{\mathcal Z}_n$ where $\shrink=\frac1{2R}$. Then the entanglement entropy of any state in $\mathcal Z$ is $O(\log R+\log\dim\mathcal Z)$.
\end{lem}
\begin{proof}
This follows from the proof of \cite[Lemma 4.7]{OTR} and the following strengthening of \cite[Lemma 4.6]{OTR}.
\end{proof}
\begin{lem}
Let $\mathcal Z$ be a subspace of $\mathcal H_1\otimes\mathcal H_2$ and suppose there exists a subspace $\mathcal V\subset\mathcal H_1$ with $\dim(\mathcal{V})=V$ which is $\delta$-viable for $\mathcal Z$. Pick any normalized state $\ket\psi\in\mathcal Z$ and write the Schmidt decomposition $\sum_{i}\sqrt{\lambda_i}\ket{x_i}\ket{y_i}$ with nonincreasing Schmidt coefficients. Then we have the tail bound $\sum_{i>V}\lambda_i\le\delta$.
\end{lem}
\begin{proof}
By the definition of $\delta$-viability, we have
\begin{equation}
	\langle z'|(P_{\mathcal{V}^\bot}\otimes I)|z'\rangle\leq\delta
\end{equation}
for all normalized state $|z'\rangle$ in $\mathcal{Z}$. In particular, this implies
\begin{equation}
	\langle z|(P_{\mathcal{V}^\bot}\otimes I)|z\rangle
	=\sum_{i}\lambda_i\langle x_i|P_{\mathcal{V}^\bot}|x_i\rangle
	=\mathrm{tr}\left(P_{\mathcal{V}^\bot}\sum_{i}\lambda_i|x_i\rangle\langle x_i|\right)
	\leq\delta.
\end{equation}
Here, $P_{\mathcal{V}^\bot}$ projects onto $\mathcal{V}^\bot$ with $\dim(\mathcal{V}^\bot)=\dim(\mathcal{H}_1)-V$. Now use Poincar\'{e}'s inequalities \cite[Corollary 4.3.39]{horn2012matrix} to conclude that
\begin{equation}
	\mathrm{tr}\left(P_{\mathcal{V}^\bot}\sum_{i}\lambda_i|x_i\rangle\langle x_i|\right)
	\geq\sum_{i=\dim(\mathcal{H}_1)-\dim(\mathcal{V}^\bot)+1}^{\dim(\mathcal{H}_1)}\lambda_i
	=\sum_{i=V+1}^{\dim(\mathcal{H}_1)}\lambda_i.
\end{equation}
This establishes the claimed bound.
\end{proof}
We apply Lemma \ref{OTRlemma} to the case of a simple ground state. Since $\operatorname{span}\ket{\psi_1}\approx_\delta\operatorname{span}\ket{\psi_2}$ where $\delta=D(\ket{\psi_1},\ket{\psi_2})^2$ we obtain
\begin{cor}\label{cor:OTR_simple}
Suppose there exist $K_1,K_2,\ldots$ such that $K_n$ is an $(\shrink^n,R^n)$-AGSP with target state $\ket{\psi_n}$ where $\shrink=\frac1{2R}$. If $\sum_{n=0}^\infty n D(\ket{\psi_n},\ket{\psi})^2=O(1)$ then the entanglement entropy of $\ket\psi$ is $O(\log R)$.
\end{cor}

The Hamiltonian $H''$ of \eqref{eq:Hpp} has spectral norm $O(s\mathcal{N}(\Lambda)+T(\Lambda))$. Under the conditions of Corollary \ref{cor:robustness}, $H''$ has a spectral gap $\Delta''=\Omega(\Delta)$, so we have
\[\Delta''/\|H''\|=\Omega\Big(\frac{\Delta}{s\mathcal{N}(\Lambda)+T(\Lambda)}\Big).\]
The Chebyshev polynomial of degree $\ell$ is bounded by $1$ on the unit interval but satisfies $T_\ell(1+\Delta)\ge\frac12(1+\sqrt{2\Delta})^\ell=\frac12\exp(\Omega(\ell\sqrt\Delta))$, so composing it with linear transformations yields a polynomial $f$ with $f_\ell(0)=1$ and $|f_\ell(\lambda)|\le2\exp(-\Omega(\ell\sqrt{\Delta/M}))$ for $\lambda\in[\Delta,M]$. Picking $M=\|H''\|$ we get that $K=f_\ell(H'')$ is an AGSP with shrinking factor
 \begin{equation}\label{eq:sf}\shrink\lesssim\exp\Big(-\Omega(\ell\sqrt{\Delta''/\| H''\|})\Big)=\exp\Big[-\Omega\Big(\frac{\ell\sqrt{\Delta}}{\sqrt{s\mathcal N+T}}\Big)\Big].\end{equation}
 Assume $\bar\lambda=O(1)$. We may pick parameters as in Eq.\ \eqref{eq:Ltchoices} so that $D(\ket{\psi_0''},\ket{\psi_0})^2\le\delta$. Since $\mathcal{N}(\Lambda)=\poly(\Lambda)$, Eq.\ \eqref{eq:Ltchoices} implies
\begin{equation}\Lambda,\mathcal{N}(\Lambda),T(\Lambda)=\poly(\Delta^{-1})\polylog(s/\delta).\end{equation}
 Eq.\ \eqref{eq:sf} then becomes
 \[\shrink(\delta)\lesssim\exp\Big[-\Omega\Big(\ell\cdot\frac{\Delta^{O(1)}}{\sqrt{s}\cdot\polylog(s/\delta)}\Big)\Big].\]

$H''$ is a spin chain with a length-$s$ segment such that each qudit on the segment has local dimension $d(\Lambda)=\Lambda^{O(1)}$. The amortization bound from the proof of \cite[Lemma 4.2]{AradKitaevLandauVazirani2013area} yields that $K$ has entanglement rank: 
\[R(\delta)=(\ell d)^{O(\ell/s+s)}\le\exp(C(\ell/s+s)\log(\ell\Lambda))\le\exp[C(\ell/s+s)\big(\log(\ell/\Delta)+\log\log(s/\delta)\big)].\]
  Let $\ell=s^2$:
 \begin{equation}R(\delta)=\exp\Big[Cs\big(\log s+\log(\Delta^{-1})+\log\log(1/\delta)\big)\Big],\qquad\sigma(\delta)=\exp\Big[-cs^{3/2}\cdot (\tfrac\Delta{\log(s/\delta)})^{O(1)}\Big].\label{Rs}\end{equation}

\begin{thm}[Area law]
\label{thm:area_law}
Under Assumptions~\ref{assumption:local_quantum_numbers} and \ref{assumption:local_quantum_number_growth} with $\bar\lambda=O(1)$, the ground state of $H$ satisfies an area law with entanglement entropy bounded by $\poly(\Delta^{-1})$.
\end{thm}
We remark that $\bar\lambda=O(1)$ is proved for the U(1) and SU(2) LGTs (Corollary \ref{cor:mean_abs_bound_LGT}), as well as for the Hubbard-Holstein model (Corollary \ref{cor:mean_abs_bound_Hubbard_Holstein}), in Section~\ref{sec:mean_abs_local_quantum_number_bound}.
\begin{proof}
Pick $\delta_k=k^{-3}$ so that $\sum k\delta_k<\infty$. Let $r$ be a parameter depending on $\Delta$ and for $k=1,2,\ldots$ consider $s=rk/\log(rk)$. Then \eqref{Rs} yields entanglement bounds and shrinking factors
 \[R_k=\exp[C_0r k\log(\Delta^{-1})],\qquad\sigma_k=\exp[-c(rk)^{3/2}\cdot (\tfrac\Delta{\log(rk)})^{O(1)}].\]
In particular we have $R_k=R^k$ where $R=R_1$. To apply Corollary \ref{cor:OTR_simple} it remains to ensure that $\sigma_k\le(2R)^{-k}$ for all $k\in\mathbb N$
 \[(2R)^k\sigma_k\le\exp\Big[k+C_0rk\log(\Delta^{-1})-c(rk)^{3/2}\cdot (\tfrac\Delta{\log(rk)})^{O(1)}\Big].\]
 So we just need the expression in the square brackets to be negative. Dividing the exponent by $rk$ we see that it suffices to pick $r$ such that for all $k$,
  \[1/r+C_0\log(\Delta^{-1})\le c\sqrt{rk}\cdot \Big(\frac\Delta{\log(rk)}\Big)^{C_1}.\]
  Pick $r=\tilde\Theta(\Delta^{-2C_1})$ to achieve this for all $k\ge1$.  
By Lemma \ref{OTRlemma} the ground state entanglement entropy is bounded by:
\[\log R=C_0r\log(\Delta^{-1})=\poly(\Delta^{^{-1}}).\]
This completes our proof.
\end{proof}

\begin{@fileswfalse}
\bibliography{bib}
\end{@fileswfalse}

\end{document}